\documentclass[conference]{ieeeconf}
\IEEEoverridecommandlockouts
\usepackage{cite}
\usepackage{amsmath,amssymb,amsfonts}
\usepackage{algorithmic}
\usepackage{graphicx}
\usepackage{textcomp}
\usepackage{xcolor}
\usepackage{gensymb}
\usepackage{svg}
\usepackage{url}
\usepackage{makecell}
\usepackage[percent]{overpic}
\usepackage{csquotes}
\usepackage{mathtools}
\usepackage{optidef}
\usepackage{float}

\def\BibTeX{{\rm B\kern-.05em{\sc i\kern-.025em b}\kern-.08em
    T\kern-.1667em\lower.7ex\hbox{E}\kern-.125emX}}

\usepackage{amsthm}

\newtheorem{prob}{Problem}
\newtheorem{thm}{Theorem}
\newtheorem{cor}{Corollary}
\newtheorem{prop}{Proposition}

\newtheorem{lem}{Lemma}
\newtheorem{defi}{Definition}

\begin{document}

\title{Maximizing Privacy in MIMO Cyber-Physical Systems \\ Using the Chapman-Robbins Bound}

\author{Rijad Alisic, Marco Molinari, Philip E. Par\'e, and Henrik Sandberg*\thanks{*Rijad Alisic, Philip E. Par\'e, and Henrik Sandberg are with the Division of Decision and Control Systems and Marco Molinari is with the Division of Applied Thermodynamics and Refrigeration at KTH Royal Institute of Technology, Sweden (e-mail: rijada, marcomo, philipar, hsan@kth.se)}}

\maketitle

\begin{abstract}
    Privacy breaches of cyber-physical systems could expose vulnerabilities to an adversary. Here, privacy leaks of step inputs to linear-time-invariant systems are mitigated through additive Gaussian noise. Fundamental lower bounds on the privacy are derived, which are based on the variance of any estimator that seeks to recreate the input. Fully private inputs are investigated and related to transmission zeros. Thereafter, a method to increase the privacy of optimal step inputs is presented and a privacy-utility trade-off bound is derived. Finally, these results are verified on data from the KTH Live-In Lab Testbed, showing good correspondence with theoretical results.
\end{abstract}

\vspace{-10pt}

\section{Introduction} \label{sec:intro}
Digitalization is rapidly transforming many aspects of society, using data collected by sensors in smart cities, manufacturing facilities and energy networks, in order to decrease costs, detect faults and improve the experience of its end-users. %Additionally, sensors help mitigate undetectable attacks against cyber-physical systems, which could drive system into undesirable states without raising any alarms~\cite{}.
%Additionally, sensors are vital in the use of protecting cyber-physical systems because they remove possible zero-dynamic attacks which could drive the system into undesirable states without the operator being able detect it~\cite{}. 
%However, introducing sensors also introduces a new attack dimension for an adversary. For example, the sensor readings could be manipulated so that the controller produces signals that could damage the system it is controlling~\cite{} through the use of false-data injections.
However, the introduction of these sensors makes eavesdrop attacks possible, breaching the confidentiality of the system that is being spied upon. %Another type of stealthy attack is one where the information that is sent by the sensors to the operator are intercepted and then end up in the wrong hands. 
Eavesdropping is difficult to detect, since it does not affect the system directly. The leaked information could be used by the attacker to figure out the structure of the underlying system and learn its weaknesses. A way to keep the information from getting into the hands of an adversary could be by means of encryption. %However, encryption might not work for the type of attacks considered here, since it is possible to figure out correlations between snippets of packets using measures based on entropy~\cite{}, even if the signals are encrpyted.
However, encryption comes with its own set of difficulties, for example, increasing time delays in data streams or increased maintenance cost due to secret key handling~\cite{Liu03selectiveencryption}. Additionally, the increased processing time could make the system more susceptible to \emph{Denial-of-Service} attacks, since it becomes easier to make the system miss its real-time computation constraints.

Instead, another defense strategy would be to introduce noise into the data stream, which makes the adversary uncertain about what the actual signal is. An example of where this approach is used is in the concept of differential privacy, which is a popular tool that hides user information in databases~\cite{Dwork_Smith_2010}. A database can release various structures of its data to anyone without explicitly revealing its individual data entries. However, an adversary could combine this information with side information to deduce the individual entries, thus breaching the privacy of the database. A differentially private database removes this possibility, for example, by corrupting the answers to queries with noise so that it is not possible to reveal individual entries with additional side information.

In dynamical systems the data has an additional component, namely time. The adversary can make use of models of the system in order to reconstruct corrupted data. One definition of privacy in this context is the estimation error of the system's internal states which is proposed in~\cite{farokhisandberg2018}. Introducing noise is a central component in that work as well and increasing the estimation error variance is related to increasing differential privacy~\cite{dwork2008}. A more direct generalization of differential privacy to dynamical systems is shown in~\cite{lenypapas2014}, where the privacy of input signals is considered.

Privacy has also been considered in the context of hypothesis testing~\cite{lioechtering2015}. An attacker considers the value of the state of the system to be different hypotheses, and uses the measurements to determine which hypothesis is true. The privacy is defined as being the type-II error of a hypothesis test, namely the probability of missing to declare that the correct hypothesis is true. 

Guided by an example of a privacy leakage scenario in a multi-residential smart building, we consider privacy to be a combination of concepts from~\cite{farokhisandberg2018} and~\cite{lenypapas2014}. A smart building uses sensors to read its current state, for example, temperature and CO$_2$ in different rooms. It also uses actuators, via a controller, to shift the building system into states which are desirable for its residents. The desirable states depend on what rooms the residents are currently in. Thus, there are streams of data inside the building which are being transmitted between the sensors, the control system, and the actuators. These data streams contain confidential information about the residents, for example, if they are home or not.

Assume now that there exists an attacker who gains access to some of the sensor-to-controller data streams, with the objective of detecting changes in the controller-to-actuator data streams, which indicates that a resident has moved from one room to another. Let the attacker know the model of the system, its initial state, and the shape of the input sequence \emph{a priori}. However, the attacker does not know at what time step the input sequence starts and, therefore, the attacker wants to estimate it. The estimation of the change time is done using a series of hypotheses tests, which translates into solving a \emph{change point problem}. Although the literature on this type of problem is extensive, there are no \emph{uniform minimum variance estimators} (UMVE) for detecting when a change occurs~\cite{lehmanncasella1998,bassevillenikiforov1993}. Instead, the adversary is forced to solve a combinatorial problem. Since there are no UMVE for the change time of a step input, we define privacy to be the lowest obtainable variance of the estimated change time using any estimator.

%Assume now that there exists an adversary who gains access to some of the sensor-to-controller signals, with the objective of detecting changes in the controller-to-actuator signal channel. Let the adversary know \emph{a priori} that there is someone in the apartment. By detecting changes in the input, they will be able to figure out when the apartment becomes empty. The information which is obtained from this breach of privacy, could then be used to conduct various attacks against the system.

%The attacker thus tries to solve a \emph{change point problem}. Although the literature on this type of problem is extensive, there are no \emph{uniform minimum variance estimators} (UMVE) for detecting when a change happens~\cite{lehmanncasella1998,bassevillenikiforov1993}. Instead, the adversary is forced to solve a combinatorial problem. Since there are no UMVE for the change time of a step input, we define privacy to be the variance of the estimated change time.

This work is an extension of a {previous paper}~\cite{alisic2020bounding_arxiv}, where the problem was restricted to single-input-single-output systems. There, an analysis of the system structure was conducted in order to figure out what structural components enable privacy leaks. For example, it was shown that a large signal-to-noise ratio (SNR) was a contributing factor. However, a defender can easily control the SNR by modifying the additive noise in the measurements. Instead, the analysis showed how slow dynamics also produce a large privacy leakage, since it allows for the adversary to collect more samples of the output signal during the change. The results of that study also revealed that unstable and integrator systems have the largest privacy leaks. %The reason for the large privacy losses was that different signals never converged to the same value, thus making it possible to sample enough data points to exactly estimate the input signals.

In this paper, we consider a similar setting, but with the extension to multiple-input-multiple-output (MIMO) systems. The inputs are assumed to be step changes, so that the system reaches some desired steady states, which implies that we only consider stable systems with no integrator dynamics appearing in the output. In this setup, fundamental lower bounds on the estimation of the change time are derived. We show that in some cases it is possible to change the inputs to satisfy the same constraints at steady state, while simultaneously increase the difficulty of estimating the change time for an adversary. %\textbf{Additionally, we assess the impact different combinations sensors has to the estimation of the change time, showing that guarding the most insecure sensor might not be the best strategy for hiding the inputs.}

%We start by showing an example of a privacy leak of a smart building, where the adversary has gained access to a couple of sensor channels. In the example, it is shown how an adversary could estimate the change of the occupancy in a room, which is modeled as an input to a dynamical system. 

In the next section, we formulate the problem by posing two main research questions. The first question asks how much information about the change time an adversary can uncover and the second asks if there is a way to use pre-existing noise in the system to change the inputs so that privacy leaks are reduced. The answers to these problems are given in Section~\ref{sec:results}, and their implications are analyzed. Finally, in Section~\ref{sec:numerics}, we return to the motivating example of an adversary eavesdropping on a smart building. There we highlight what the theory predicts about an adversary’s capabilities of breaching the privacy of the residents and how a defender could design input signals that minimize the privacy leakage.

\vspace{-5pt}

\section{Problem Formulation}\label{sec:probform}

%Consider a linear time invariant system where the measurements are corrupted by a zero-mean Gaussian signal, $e_k \in \mathcal{N}(0,\Sigma_e)$, then the system model $M$ can be written as
%\begin{equation} \label{eq:model}
%M : \begin{cases}
%    x_{k+1} = Ax_k+Bu_k  \\
%    \hspace{11pt} y_k = Cx_k  + e_k,
%\end{cases}  
%\end{equation}
%where $x_{k+1} \in \mathbb{R}^n$, $y_k \in \mathbb{R}^m$ and $u_k \in \mathbb{R}^p$. The system matrices, $A\in \mathbb{R}^{n \times n}$, $B\in \mathbb{R}^{n \times p}$ and $C \in \mathbb{R}^{m \times n}$ define the system model, $M$ together with the disturbance covariances. Denote the sequence of measurements and inputs as $Y = (y_k)_{k=0}^N$ and $U = (u_k)_{k=0}^N$, respectively. Each input vector, $u_k$ is assumed to be a step, such that $u_k=u \Theta (k-k^*)$, where $\Theta(k)$ is the unit step function and $u$ is the input direction.

%Now, assume that the attacker listens in on the communication channel between the sensors and the controller. Let the attacker have a linear time invariant model of the system, $M$,

Consider a linear time invariant system where the measurements are corrupted by a zero-mean, {stationary, white Gaussian signal, $e_k$, with covariance $\mathbb{E}\left[e_k e_k^\top\right] = \Sigma_e$, $\forall k$}. Then, the system model $M$ can be written as:
\begin{equation} \label{eq:model}
M : \begin{cases}
    x_{k+1} = Ax_k+Bu_k  \\
    \hspace{11pt} y_k = Cx_k  + e_k,
\end{cases}  
\end{equation}
where $x_{k} \in \mathbb{R}^n$, $y_k \in \mathbb{R}^m$, $e_k \in \mathbb{R}^{m}$ and $u_k \in \mathbb{R}^p$. The system matrices $A\in \mathbb{R}^{n \times n}$, $B\in \mathbb{R}^{n \times p}$ and $C \in \mathbb{R}^{m \times n}$ together with the noise model {for $e_k$} define $M$. Denote the sequence of outputs and inputs as {$Y = \{y_k\}_{k=0}^N$ and ${U = \{u_k\}_{k=0}^N}$,} respectively. %The input sequence, $U$, is assumed to be a step, such that $u_k=u \Theta (k-k^*)$, where $\Theta(k)$ is the unit step function, $\Vert u \Vert_2$ is the size of the input, and $u/ \Vert u\Vert$ is its direction.
The input sequence, $U$, is assumed to be a step,
\begin{equation*}
    u_k = \begin{cases} 0, \quad \text{for } k<k^*, \\ u, \quad \text{for } k\geq k^*, \end{cases}
\end{equation*}
where $\Vert u \Vert_2$ is the size of the input, and {$u/ \Vert u\Vert_2$ is its direction.}

The objective of the attacker, is to estimate the change time, $k^*$, using the model $M$ and the measurements $Y$. A defender's main purpose is then to make it as difficult as possible for the attacker to obtain their goal. Motivated by the attacker's goal, we define privacy in the following manner:
\begin{defi}\label{def:priv}
Consider an estimator of the change time $k^*$ for the inputs $U=\{u_k\}_{k=0}^N$, which are fed through system $M$ in~\eqref{eq:model}. Denote the estimator of $k^*$ by $\psi$, {which has a bias that is bounded}. We define the privacy of system $M$ to be the lowest achievable variance of the estimated change time,
\begin{equation*}
    \min \limits_{\psi} \left( \mathrm{Var}(\psi|k^*)\right).
\end{equation*}
\end{defi}
This definition of privacy is general and the defender may consider estimators which take very complex information into account. The problem of interest in this paper, however, is to calculate the privacy of system $M$, conditioned on the type of estimators that the attacker can produce.

%estimators that the attacker can produce

%The example in Section~\ref{sec:privleak} provides a description of the attacker's capabilities, which we can use to ask questions about how private a system is, when it is observed by a such adversary: 

%Estimating when the input sequence $U$ changes between two values is limited by how much the input affects the output and how well stochastic perturbations hide this impact. %When the input is completely unknown, it is possible to derive fundamental bounds on the estimation quality based on the Fisher information as in. However, In the application that is considered here, The attacker knows the input sequence \textit{a priori} via the system model~\eqref{eq:model}. The inputs are constant over long periods of time and quick step changes only being made occasionally. Additionally, the attacker is not particularly interested in the actual value of the signal, they only want to know when it changed. With these structural factors in mind the following question is posed,

\begin{prob} \label{prob:problem}
Let an estimator of the change time $k^*$ in $U$, denoted by $\psi_u(Y,M)$, have access to the model~\eqref{eq:model} and the measurements $Y$ of length $N$ such that $N \geq k^*$. What is the minimum variance that any such estimator can achieve? \end{prob}

In physical systems, measurement noise is typically present or, alternatively, injecting noise into measurements might come with some costs. Consequently, a defense strategy might  include ways to choose the input such that the existing noise is used to hide the input. Therefore, we pose the following question:
\begin{prob}\label{prob:prob2}
Consider an estimator of $k^*$, $\psi(Y,M)$, that has access to the model~\eqref{eq:model} and the measurements $Y$ of length $N$ such that $N \geq k^*$. Is it possible to choose $U$ so that the lower bound on the variance of any such estimator is increased?
\end{prob}

The answer to these questions will show what structures in the model, $M$, expose the change time, $k^*$, to an adversary. The answer also provides the defender with information about how to design their system so that estimating the change time becomes as difficult as possible.  Although any level of privacy can be achieved by injecting enough noise into the system, additional noise also degrades the controller performance. If the controller aims to minimize a cost function, then the actual cost increases when noise is added. It is therefore important that the noise which is already present is used to the fullest extent, which could be done by placing the sensors strategically or by designing controllers that make multiple actuators cooperate so that a particular change is more difficult to estimate. {In this paper, we will assume that the defender knows the noise model \textit{a priori}, which might not always be true for real systems.}

\vspace{-5pt}

\section{Main Results}\label{sec:results}
The Cram\'er-Rao lower bound~\cite{Jansen2011} is typically used to answer questions like Problem~\ref{prob:problem}. A difficulty here is that the Cram\'er-Rao lower bound is only defined for continuous parameters, whereas $k^*$ takes discrete values. Therefore, a more general result is required in order to answer Problem~\ref{prob:problem}. Such a result is made possible by the Chapman-Robbins (CR) bound~\cite{chapmanrobbins1951}. With the CR-bound, it is possible to show what structures in the model $M$ expose the change time to an adversary. {Therefore, subsequently we will use this information to find an input that increases the smallest possible variance of the estimate, defined as $\mathcal{B}_u(M)$.} %The theorem will serve as a basis for the subsequent results and is therefore the main result of this paper.

\begin{thm} \label{thm:main}
Consider any estimator of the change time $k^*$ in the input sequence $U$. Denote the estimator by $\psi(Y,M)$ with bias $g(k^*)$, where $M$ is a MIMO-system. Then \vspace{-1pt}
\begin{equation} \label{eq:bound}
    \mathrm{Var}(\psi_u(Y,M)|k^*) \geq \mathcal{B}_u\left(M \right),
\end{equation}
\vspace{-10pt}where,

\begin{equation*}
    \mathcal{B}_u\left(M \right) \coloneqq \max \limits_{\tau} \frac{(\tau + g(k^* +\tau) - g(k^*))^2}{\mathrm{e}^{{ u^\top\mathcal{S}(\tau,M)u }}-1},
\end{equation*}
for $\tau \in \{1, \,  \dots, \, N-k^*\}$. Here,
\begin{equation} \label{eq:sum}
    \mathcal{S}(\tau,M) = \sum \limits_{k=k^* +1}^N  \left(C\tilde A(k,\tau) B \right)^\top \Sigma^{-1}_eC\tilde A(k,\tau)B,
\end{equation}
\vspace{-10pt}where, 
\begin{equation} \label{eq:t1e3}
    \tilde A(k,\tau)= \left(\sum \limits_{l=k^*}^{\mathrm{min}(k^* + \tau -1,k-1)} A^{k-1-l}\right).
\end{equation}
%and,
%\begin{equation} \label{eq:noisemat}
%    \Sigma(k) = \Sigma_e+C\left(\sum \limits_{l=0}^{k-1} A^l \right)\Sigma_w \left(\sum \limits_{l=0}^{k-1} A^l \right)^\top C^\top.
%\end{equation}
\end{thm}
\begin{proof}
See the Appendix.
\end{proof}

%In this section, we aim to answer Problems~\ref{prob:problem} and~\ref{prob:prob2}. The main theorem is stated first, which provides a way to quantify privacy as it is defined in Definition~\ref{def:priv}. This theorem will serve as a basis for finding input directions which are hardest to estimate, which we call \textit{most private input directions}. In the following subsections, systems containing input directions that are impossible to estimate are analyzed first. Thereafter, a method to classify input directions based on their privacy level is presented and how one can use them to improve privacy is shown.

Much attention will be given to $\mathcal{S}(\tau,M)$ in the subsequent sections, since it determines the smallest possible variance of any estimator, $\mathcal{B}_u(M)$. It will be assumed that the estimator is unbiased, so that $g(k^*)=0$. The minimum variance in~\eqref{eq:bound} depends explicitly on $u$, which implies that different inputs provide different levels of privacy. 

\iffalse
\begin{defi}\label{def:privdir}
Let $\underline \lambda_\tau(\mathcal{S})$ be the smallest eigenvalue of $\mathcal{S}(\tau,M)$. The most private input direction is defined to be the eigenvector $u_*$, corresponding to the $\underline\lambda_{\tau^*}$ eigenvalue, where
\begin{equation*}
    \tau^* = \arg \max \limits_\tau \frac{\tau^2}{\mathrm{e}^{\underline\lambda_\tau(\mathcal{S})}-1}.
\end{equation*} 
We say that $u_*$ is a fully private input direction if $ \underline \lambda_{\tau^*}(\mathcal{S})=0$.
\end{defi}
\fi
{
\begin{defi}\label{def:privdir_1}
Let $\underline \lambda_\tau(\mathcal{S})$ be the smallest eigenvalue of $\mathcal{S}(\tau,M)$. The most private input direction is defined to be the eigenvector $u_*$, corresponding to the $\underline\lambda_{\tau^*}$ eigenvalue, where $\tau^*$ maximizes~\eqref{eq:bound}.

\iffalse
\begin{equation*}
    \tau^* = \arg \max \limits_\tau \frac{\tau^2}{\mathrm{e}^{\underline\lambda_\tau(\mathcal{S})}-1}.
\end{equation*} 
\fi
\end{defi}

Furthermore, it may be impossible to estimate the change time for some particular directions of $u$. These directions are given by the following definition.

\begin{defi}\label{def:privdir_2}
We say that $u_*$ is a fully private input direction if $ \underline \lambda_{\tau^*}(\mathcal{S})=0$.
\end{defi}
}
Definition~\ref{def:privdir_1} is justified by the following proposition, which states that step-changes in the direction of $u_*$ provide the most privacy to a system.

\begin{prop}\label{prop:most1private}
The smallest variance for estimating change time of the step, $k^*$, is maximized in the direction of $u_*$.
\end{prop}
\begin{proof}
Minimizing $u^\top\mathcal{S}(\tau,M)u$ for $\tau=\tau^*$ maximizes the bound in~\eqref{eq:bound}.
\end{proof} \vspace{-8pt} 
\noindent 
Proposition~\ref{prop:most1private} provides a method to find the most private input directions. Therefore, it also gives an affirmative answer to the question that is posed in Problem~\ref{prob:prob2}. 

Notice that so far in the discussion, we have not mentioned the size of the inputs. Theorem~\ref{thm:main} states that a larger $\Vert u \Vert$, implies a smaller $\mathcal{B}_u(M)$ and thus less privacy. Therefore, it is fully possible that the change time in the the most private input direction, $u_*$, is easier to estimate than some $u$, which is not in that direction, if $\Vert u_* \Vert>\Vert u \Vert$. This is related to the SNR, which was treated in~\cite{alisic2020bounding_arxiv}.

\subsection{Fully Private Input Directions}

Fully private input directions have an interesting property, namely that they make $\mathcal{B}_{u}=\infty$. Therefore, it is impossible to estimate when step-changes in this direction occur, since any unbiased estimator of $k^*$ will have an infinite variance. Methods for finding fully private input directions, will be presented in this subsection.

%The lower bound that is given in Theorem~\ref{thm:main} can be used to analyze the privacy properties of the system. For example, a necessary and sufficient condition for a system to have a fully private input direction is that there exists an input direction which does not affect the output, which is proved in the following theorem.

\begin{thm} \label{thm:privin}
A fully private input direction exists if and only if
\begin{equation}\label{eq:t2e1}
    \mathrm{rank}(\mathcal{O}) < p,
\end{equation}
where
\begin{equation}\label{eq:t2e2}
        \mathcal{O} = \begin{bmatrix} C B \\ C A B \\ C A^2 B \\ \vdots \\ C A^{N-1} B \end{bmatrix}.
\end{equation}
\end{thm}

%Before we prove this theorem, let us state the following lemma:

Before we prove this theorem, we need to relate the null space of~\eqref{eq:sum} to~\eqref{eq:t2e2}. 

\begin{lem}\label{lem:inputobservability}
The following holds for any $\tau,\tilde \tau \in \mathbb{N}$,
\begin{equation*}
    \mathcal{O}(\tilde \tau)u=0 \iff \mathcal{O}(\tau)u=0,
\end{equation*}
where
\begin{equation*}
   \mathcal{O}(\tau)= \begin{bmatrix} C\tilde A(0,\tau) B \\ C\tilde A(1,\tau) B \\ C\tilde A(2,\tau) B \\ \vdots \\ C\tilde A(N-1,\tau) B \end{bmatrix}.
\end{equation*}
\end{lem}
\begin{proof}
Without loss of generality, we can set $\tilde \tau=1$, where $\mathcal{O}(1)=\mathcal{O}$ in~\eqref{eq:t2e2}. The proof for
\begin{equation*}
    \mathcal{O}u=0 \Rightarrow \mathcal{O}(\tau)u=0,
\end{equation*}
is obtained by using~\eqref{eq:t1e3}, which shows that $C\tilde A (k,\tau)B$ is a linear combination of $CA^lB$, $\forall l \in \{1, \dots, N-1 \}$. The proof for
\begin{equation*}
    \mathcal{O}u=0 \Leftarrow \mathcal{O}(\tau)u=0,
\end{equation*}
follows from by rewriting~\eqref{eq:t1e3} as,
\begin{equation*}
    \tilde A(k,\tau)  = \begin{cases}
    I, &  k=0, \\
    A^k+\tilde A(k-1,\tau), & 1\leq k \leq \tau-1 \\
    A^{k+1-\tau} \tilde A(\tau-1,\tau), & \tau \leq k \leq N-1.
    \end{cases}
\end{equation*}
Thus each $C \tilde A (k,\tau)B$ is a linear combination of $CA^kB$ and $C \tilde A (l,\tau)B$, for all $l<k$, which implies $$ 0=\mathcal{O}(\tau)u \Rightarrow \mathcal{O}(\tau)u = \mathcal{O}u.$$ \end{proof}

Let us return to the proof of Theorem~\ref{thm:privin}.

\begin{proof}[Proof of Theorem~\ref{thm:privin}]
If~\eqref{eq:t2e1} holds, then by Lemma~\ref{lem:inputobservability} there is a $u_*$ such that
\begin{equation*}
    \mathcal{O}u_*=\mathcal{O}(\tau)u_*=0,
\end{equation*}
which implies that $u_*^\top\mathcal{S}(\tau,M)u_*=0$ for all $\tau$, making $\mathcal{B}_{u_*}(M) = \infty$.

If $\mathcal{B}_{u_*}(M) = \infty$, then there is a $\tau^*$ such that
\begin{equation} \label{eq:p2e1}
    u_*^\top\mathcal{S}( \tau^*,M)u_*=0
\end{equation}
for some $u_*$. Since $\Sigma_e$ is positive definite,~\eqref{eq:p2e1} implies that
\begin{equation*}
    C \tilde A(k,\tau^*) Bu_* = 0, \; \forall k.
\end{equation*}
Finally, the last step is obtained by applying Lemma~\ref{lem:inputobservability}.\end{proof}

Theorem~\ref{thm:privin} states that there exists a fully private input direction if the input-observability matrix, $\mathcal{O}$, is rank deficient. If it is rank deficient, then the fully private input direction is in the null space of $\mathcal{O}$. Inputs in this direction do not affect the output, which is similar to what inputs that excite zero dynamics do. For sample horizons that are large enough, namely $N-k^* > n$, fully private input directions are a special case of these. Recall the definition of a transmission zero:

\begin{defi}
A zero, $z_0$, is a complex number that makes the Rosenbrock system matrix rank deficient,
\begin{equation*}
    \mathrm{rank}\left(\begin{bmatrix} A-Iz_0 & B \\ C & 0 \end{bmatrix}\right) < m+n, \quad \text{where } m\geq p.
\end{equation*}
We denote $x^0$ as the zero-state direction and $u^0$ as the corresponding zero-input direction, where
\begin{equation}\label{eq:zerodir}
    0=\begin{bmatrix} A-Iz_0 & B \\ C & 0 \end{bmatrix} \begin{bmatrix} x^0 \\ u^0 \end{bmatrix}.
\end{equation}
\end{defi}

A simpler way to determine if an input direction is fully private is by checking the zero-state direction.
\begin{cor}
Let the measurement horizon satisfy $N-k^* > n$. Then an input direction is fully private if and only if it is a zero-input direction, $u^0$, and
\begin{equation}\label{eq:observabilityzero}
    \begin{bmatrix} C \\ CA \\ \vdots \\CA^{n-1} \end{bmatrix}x^0 = 0.
\end{equation}
\end{cor}
\begin{proof}
Note that, using the Cayley-Hamilton theorem,~\eqref{eq:zerodir} is equivalent to,
\begin{equation}\label{eq:zerosfull}
    \begin{bmatrix} C \\ CA \\ \vdots \\ CA^{n-1} \end{bmatrix}x^0 + \begin{bmatrix} 0 \\ CB  \\ \vdots \\ \sum \limits_{k=0}^{n-2} z_0^{n-2-k}CA^{k}B \end{bmatrix}u^0 = 0.
\end{equation}

If the input direction is fully private, then the second term in~\eqref{eq:zerosfull} is zero. One may then choose $x^0=0$ which gives that~\eqref{eq:zerosfull} is zero. If the input is a zero-input direction and~\eqref{eq:observabilityzero} holds, then~\eqref{eq:zerosfull} being zero implies $\mathcal{O}u^0=0$.
\end{proof}

Thus, if the number of samples, $N$, is larger than the number of states, $n$, one may study the zeros and the respective zero-state direction. If the zero-state direction is in the observable space of the system $M$, then there exists a fully private input direction and it is parallel to the zero-input.

\subsection{Privacy-Utility Trade-off} \label{sec:regul}

It might not always be desirable to use the most private input direction, especially if it is fully private. A system designer might require that some amount information about the inputs leaks in order to guarantee control performance, run diagnostics or to detect attacks. Therefore, it is useful to add something in between a ``fully private'' and a ``non-private'' input. In order to do this, a type of privacy measure is required.
\begin{defi}
    We say that the input direction $u_*$ for model $M$, is more private than $u$ if,
    \begin{equation*}
        \mathcal{B}_{u_*}\left(M \right)>\mathcal{B}_u\left(M \right), \quad \text{for } \Vert u_* \Vert = \Vert u \Vert.
    \end{equation*}
\end{defi}

Since the defender knows how the inputs change between different modes of operation, it can tailor the additive measurement noise by using Theorem~\ref{thm:main} in order to hide these changes. However, there could be a cost associated with generating noise, for example, if a battery with finite energy is used to perturb the measurements. If the system already has inherent measurement noise, then the designer could create controllers that use the existing noise to mask the input. %To this end, the following proposition is of interest,
% \begin{prop}\label{prop:mostprivate}
% Let $\underline \lambda_\tau(\mathcal{S})$ be the smallest eigenvalue of $\mathcal{S}(\tau,M)$. The most private input direction is given by the eigenvector, $u_*$, corresponding to the $\underline\lambda_{\tau^*}$ eigenvalue, where
% \begin{equation*}
%     \tau^* = \argmin \limits_\tau\underline\lambda_\tau(\mathcal{S}).
% \end{equation*} The variance~\eqref{eq:bound} for estimating time-step changes is maximized in this direction.
% \end{prop}
% \begin{proof}
% Minimizing $u^\top\mathcal{S}(\tau,M)u$ for $\tau=\tau^*$ maximizes the bound in~\eqref{eq:bound}.
% \end{proof}

%The designer of a system usually wants to minimize some cost function which could be related to energy usage or user discomfort. However this might produce input directions which leak a lot of information. Due to Proposition~\ref{prop:mostprivate}, one may use the matrix $\mathcal{S}(\tau^*,M)$ as a regularizer in the cost function when calculating what inputs should be taken in order to push the optimal inputs towards the most private ones. Regularizing terms in the cost function usually increases cost. Therefore, it is vital that the impact on the cost is captured and that some guarantees are given. For example, 

Assume that a designer of system $M$ wishes to create a step input which minimizes the cost at steady state,
\begin{mini}
{x,u}{J(x,u)=x^\top Q x + u^\top R u}{}{}
\label{mini:cost}
\addConstraint{x}{=Ax+Bu}
\addConstraint{C_1x}{=r,}
\end{mini}
where $C_1 \in \mathbb{R}^{q \times n }$, $q<p$. %The cost function states that there is a cost to keeping the states $x$ different from zero, and the second term gives the cost of the inputs. 
Both $Q\in \mathbb{R}^{n \times n}$ and $R\in \mathbb{R}^{p \times p}$ are positive definite. The first constraint states that the system is at steady state, whereas the second constraint represents the desired reference value of some linear combinations of the states. The number of rows in $C_1$ is smaller than the number of inputs in order to ensure non-trivial solutions to the program. %Additionally, due to the summation in~\eqref{eq:sum}, it is better to move from one input to the other as quick as possible, minimizing the time that two possible output signals are different. Therefore, it is assumed that the inputs will change as steps between the optimal steady state inputs.

Assume now that the designer is willing to pay a price in the optimal cost in order to increase the privacy. They can do so by adding a regularizer to the optimization in the following way,
\begin{mini}|c|
{x,u}{J_p(x,u)=x^\top Q x + u^\top R u + \mu u^\top \mathcal{S}(\tau^*) u}{}{}
\label{mini:regularized}
\addConstraint{x}{=Ax+Bu,}
\addConstraint{C_1x}{=r,}
\end{mini}
where $\tau^*$ is the $\tau$ that maximizes~\eqref{eq:bound} for the input which minimizes~\eqref{mini:cost}. Additionally, let $(x_*,u_*)$ and $(x_p,u_p)$ denote the solutions to~\eqref{mini:cost} and~\eqref{mini:regularized}, respectively. The regularization parameter $\mu$ determines how large of a cost increase is tolerated in order to improve privacy. It is important to gain some intuition on how to choose this parameter and what guarantees can be given for a specific value of the parameter.
\begin{thm} \label{thm:privutiltrade}
The nominal privacy gain, $\delta$, is lower bounded by the utility loss, $\varepsilon \geq 0$ and upper bounded by the current privacy cost,
\begin{equation} \label{eq:t3e1}
    \mu u_*^\top\mathcal{S}(\tau^*)u_* \geq \mu \delta \geq \varepsilon,
\end{equation}
where,
\begin{equation*}
    \varepsilon = J(x_p,u_p)-J(x_*,u_*),
\end{equation*}
and
\begin{equation*}
    \delta = u_*^\top\mathcal{S}(\tau^*)u_*-u_p^\top\mathcal{S}(\tau^*)u_p.
\end{equation*}
\end{thm}
\begin{proof}
The first inequality is trivial, due to the definition of $\delta$. The second inequality is obtained by comparing the regularized costs, $J_p$, and rearranging the terms:
\begin{equation*}
\begin{aligned}
    & J_p(x_*,u_*)  \geq J_p(x_p,u_p)  \\ 
    & \iff x_*^\top Q x_*  + u_*^\top Ru_* + \mu u_*^\top \mathcal{S}(\tau^*)u_* \\ 
    & \geq x_p^\top Q x_p  + u_p^\top R u_p + \mu u_p^\top \mathcal{S}(\tau^*) u_p \\ 
    &  \iff \mu \left( u_*^\top \mathcal{S}(\tau^*)u_* - u_p^\top \mathcal{S}(\tau^*) u_p \right) \\ &\geq
     x_p^\top Q x_p  + u_p^\top R u_p - x_*^\top Q x_*  - u_*^\top Ru_*.
\end{aligned}
\end{equation*}
The last two lines is the second inequality stated in~\eqref{eq:t3e1}. \end{proof}

One may use the bounds in Theorem~\ref{thm:privutiltrade} to choose a $\mu$ which fulfills some guarantees. For example, if a maximum cost increase, $\bar \varepsilon$, is tolerated, then choosing the following $\mu$ ensures that the cost increase is upper bounded,
\begin{equation*}
    \mu = \frac{\bar \varepsilon}{u_*^\top \mathcal{S}(\tau^*)u_*} \Rightarrow \bar \varepsilon \geq \varepsilon.
\end{equation*}
This bound is not tight, which means that the designer can tune the parameter $\mu$ in order to increase privacy and thus increase the nominal $\varepsilon$ until $\bar \varepsilon$ is reached, if the initial $\mu$ does not provide sufficient privacy.

The second inequality in Theorem~\ref{thm:privutiltrade} gives an interpretation of what the regularization parameter does. By rewriting the inequality in the following manner,
\begin{equation}\label{eq:tradeoff}
    \mu \geq \frac{\varepsilon}{\delta},
\end{equation}
one may see that $\mu$ limits the maximum privacy-utility trade-off. When the designer chooses a specific $\mu$, they set the maximum tolerable utility loss per privacy unit that is gained. The bound also states that increasing the utility cost will increase the privacy as well, and conversely, if there is no privacy gain, then there will be no increase in the utility cost.

\begin{figure}
    \centering
    \includegraphics[width=1\linewidth,trim={0cm 8.2cm 0cm 8.2cm},clip]{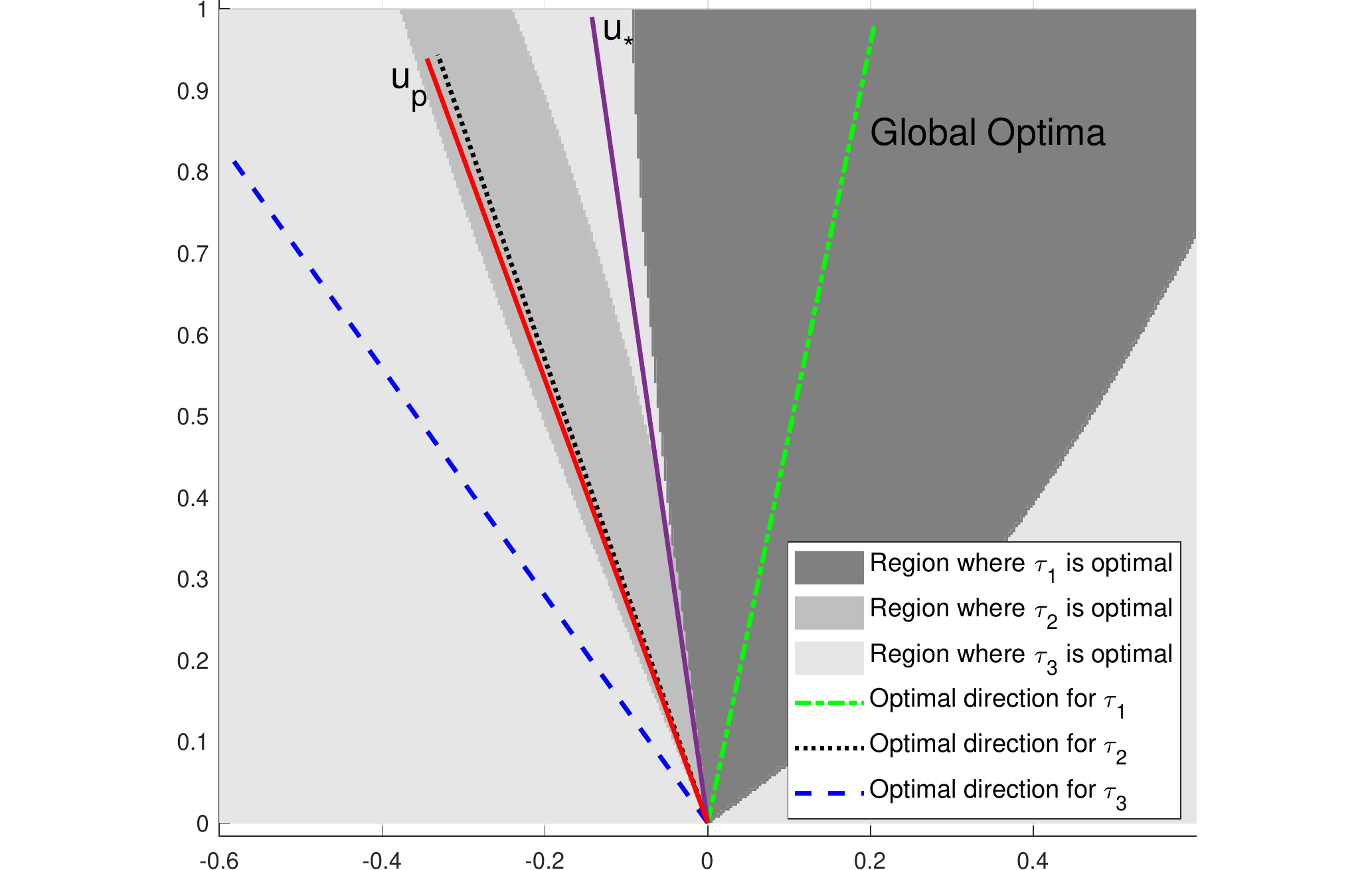}
    \vspace{-10pt}
    \caption{The figure shows how $\tau^*$ in~\eqref{eq:bound} changes as a function of $u$ for a system with five states and two inputs.}
    \vspace{0pt}
    \label{fig:optcost}
\end{figure}

{
Notice that the bound in~\eqref{eq:tradeoff} is only tight if $\tau^*$ minimizes~\eqref{eq:bound} for both $u_*$ and $u_p$. Consider the simple illustrating example which is shown in Fig.~\ref{fig:optcost}, where $\tau_3$ minimizes~\eqref{eq:bound} for $u_*$. The regularizer in~\eqref{mini:regularized} will push the solution $u_p$ towards the dashed, blue line, which gives maximum privacy for $\tau_3$. For some values of $\mu$, the solution $u_p$ might end up in a region that is maximized by $\tau_2$. In this case, $\mu$ can be interpreted as a looser bound on the maximum privacy-utility trade-off since the privacy gain is larger than what is captured by $\delta$.}

\vspace{-10pt}

\section{Numerical Results}\label{sec:numerics}
Let us return to the motivating example which was described in Section~\ref{sec:intro}, where an adversary tries to detect changes in the occupancy of different rooms in an apartment. The data for this example was taken from the IDA ICE 4.8 simulator of the Live-In Lab, KTH Testbed~\cite{liveinlabkth,idasimulator}. The Live-In Lab is a physical multi-residential building outfitted with sensors in every apartment that measure temperature, relative humidity, and CO$_2$. The KTH Testbed is a modifiable subsection of the Live-In Lab, containing 120 square metres of living space that are split up into four apartments in this example, which can be seen in Figure~\ref{fig:map}.
%{\color{blue}[why not have the graphic of the rooms here?]}

\begin{figure}
    \centering
    \includegraphics[width=1\linewidth,trim={7.6cm 4.7cm 7.6cm 3.2cm},clip]{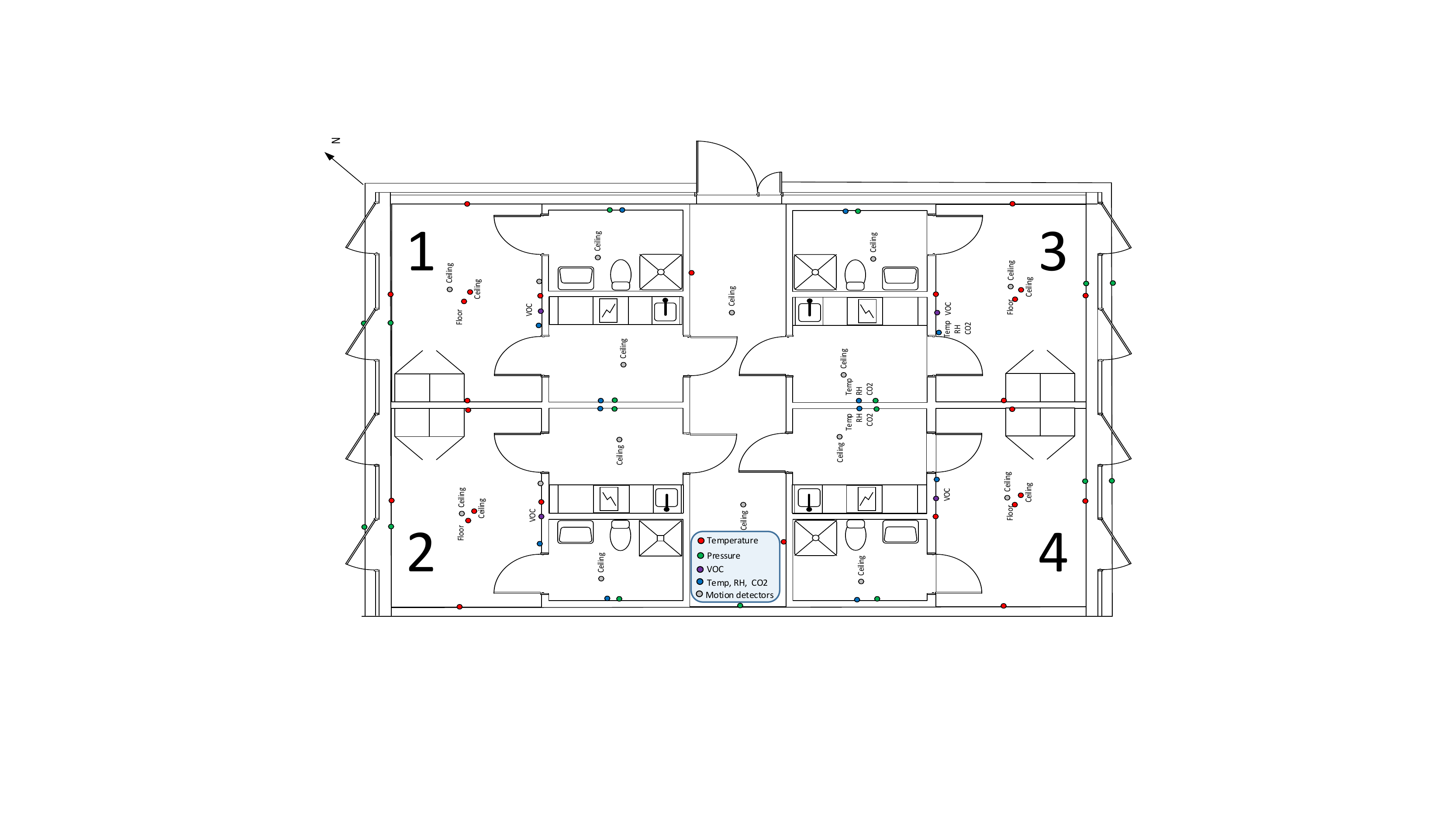}
    \vspace{-10pt}
    \caption{An overview of the apartments in the KTH Testbed.}
    \vspace{0pt}
    \label{fig:map}
\end{figure}

The adversary is assumed to have obtained the same model of the system as the defender, for example, through studying input-output data of similar apartments. Now, let the adversary eavesdrop on the system by sampling temperature and CO$_2$ measurements of the different rooms in the apartment every 9 minutes. It then estimates the input change by using a Moving Horizon Estimator~\cite{rawlings2013}. At the same time, the defender injects noise into the measurements in order to increase the variance of the adversary's estimation of the change time. The variance in the CO$_2$ sensor is four times larger than the temperature sensor, due to CO$_2$ being a much stronger indicator of occupancy.

\subsection{Detecting Occupancy Changes in Apartments}

%Let us return to the privacy leakage example in Section~\ref{sec:privleak}. There one could see that the attacker was able to obtain a fairly good estimate of when the occupancy of an apartment changes. Now, let the defender inject additive noise into the measurement channels. In light of Problem~\ref{prob:problem}, one can ask how much privacy is induced by this action in the sense of how much the variance of the estimated change time increases. 

%In Table~\ref{tab:firstex}, the theoretical lower bound on the variance, as given by $\mathcal{B}_u(M)$, together with the empirical variance, $\hat V$, for \textbf{100 trials} of occupancy change in the different rooms is presented.

\begin{table}
    \centering
    \caption{The table shows the lower bound, $\mathcal{B}_u$, and the empirical variance $\hat V$ in different rooms. The last column shows the input projected onto the most private input direction, $u_*$.}
    \vspace{0pt}
    \begin{tabular}{|c|c|c|c|} \hline
       Room &  $\mathcal{B}_u$ [min$^2$] & $\hat V$ [min$^2$] & $u \cdot u_*$ \\ \hline
         Living Room & 169  & 1570 & 1 \\
         Kitchen & 5.67 & 277 & 0.002 \\ 
         Bathroom & 18.6 & 145 & 0 \\ \hline
    \end{tabular}
    \label{tab:firstex}
\end{table}

Table~\ref{tab:firstex} shows the variance of the estimated change time for different rooms in the apartment. In the last column, the projection of the input to the system onto the most private input direction, $u_*$, is shown. Although the empirical variance, $\hat V$, increases as the input $u$ becomes more parallel to $u_*$, the same is not true for the theoretical lower bound, $\mathcal{B}_u$. The lower bound $\mathcal{B}_u$ in Table~\ref{tab:firstex} is largest for the input in the most private direction, which verifies Proposition~\ref{prop:most1private}, however, the input which is perpendicular to $u_*$ does not produce the lowest theoretical bounds. This discrepancy is explained by the non-convexity of~\eqref{eq:bound}. As $u$ is slowly rotated towards $u_*$, a different value of $\tau$ might become the minimizer of~\eqref{eq:bound}, with a different $\mathcal{S}(\tau,M)$. Because of this change, the input might pass a couple of local minima during the rotation.
%Although there seems to be a positive correlation between the direction of $u$ and the empirical variance $\hat V$, one may see that the same correlation is not present in the theoretical lower bound, $\mathcal{B}_u$. The bound $\mathcal{B}_u$ is the lowest for the input direction which is not fully parallel to the most private input direction. 
%{However, this discrepancy is} expected due the same reason {\color{blue}[what reason is that? ... this sentence is really hard to follow maybe explain why the bound is not tight first and then explain this point second...]} that the bound~\eqref{eq:tradeoff} is not tight. Due {\color{blue}[two ``due''s in back to back sentences is a little weird]} to different input direction having different $\tau$ as minimizers in~\eqref{eq:bound}, one might encounter local minima of privacy while the input is rotated towards the most private input. 
Therefore, the high empirical variance $\hat V$ of the change time in the kitchen might be due to the sub-optimality of the estimator which is used.

%increases when the variance of the measurement noise, $\sigma^2$, increases, which is expected due to the fact that information about the signal is destroyed when stochasticity is added to it. Additionally, the empirical variance, $\hat V$ of the change time, $\hat V$ is larger than the theoretical lower limit. Since it is a lower bound, this verifies Theorem~\ref{thm:main}. \textbf{The privacy cost, $u^\top \mathcal{S}u$, which is used to calculate the theoretical lower limit, is also shown.}

%One might ask if the input direction which is used by the model is the most private one, or if it is possible to combine the different input signals such that the privacy leak is reduced. By analyzing the privacy matrix for the apartments by extracting the most private input direction which is provided in Proposition~\ref{prop:mostprivate}, one will obtain a combination of inputs that optimally hide the change time $k^*$. However, inputs in this direction will typically only reach a particular combination of steady state values. If this combination ensures that all states are in the comfort zone of the user, then it is possible to use this input direction in order to produce an optimally private change of the steady state. However, this approach does not take the energy costs of the system into account, which is considered in the next section.

\subsection{Private Steady State}

%{\color{blue}[this section is pretty short; perhaps this and the previous section should be combined into one `Implementation" or `Simulations" Section and split into subsections...]}

Let us now consider the case where the different rooms in one apartment cooperate in order to reduce the privacy leak. Consider a user entering their apartment. Then, instead of only heating the room that the user enters, the building could increase the heat production in some of the other rooms as well, thus obfuscating the attacker's estimation of the change time. This control input is obtained by solving~\eqref{mini:regularized}. In Table~\ref{tab:secondex}, the impact of changing the regularization parameter is shown. In the first row, the controller aims to only minimize the utility cost, whereas in the other two rows, the controller signal aims to both minimize the utility cost and to reduce the privacy leak. The private controllers are obtained for two different values of $\mu$, under the same measurement noise covariance. Higher values of $\mu$ did not produce any noticeable improvements in the privacy. One may see that the privacy-utility trade-off bound, given by~\eqref{eq:tradeoff}, holds for all instances. As discussed in Section~\ref{sec:regul}, the bound on the privacy-utility trade-off becomes tighter for smaller $\mu$. By increasing the {$\mu$} parameter, a larger trade-off is allowed. Additionally, one may see that both the theoretical and empirical variance, $\mathcal{B}_u$ and $\hat V$, increase, as $\mu$ increases as well.%, which verifies Proposition~\ref{prop:ensurepriv}.

\begin{table}
    \centering
    \caption{The table shows how changing the maximum privacy-utility trade-off, $\mu$, affects $\frac{\epsilon}{\delta}$, $\mathcal{B}_u$ and $\hat V$.}
    \vspace{-5pt}
    \begin{tabular}{|c|c|c|c|} \hline
        $\mu$ & $\frac{\epsilon}{\delta}$ &$\mathcal{B}_{u}$ [min$^2$] & $ \hat V ^2$ [min$^2$]  \\ \hline
        0 & - & 201 & 6 470 \\
        $2 \cdot 10^{-9}$ & $0.4 \cdot 10^{-9}$ & 489 & 45 500\\
        $4 \cdot 10^{-9}$& $1.13 \cdot 10^{-9}$ & 653 &  732 000 \\\hline
    \end{tabular}
    \label{tab:secondex}
    \vspace{0pt}
\end{table}

\vspace{-10pt}
\section{Conclusions}
\vspace{-5pt}
This paper shows how Gaussian noise can be used to hide the input changes to a multi-input-multi-output system. Specifically, the relation between private inputs and transmission zeros was analysed. %\textbf{and it was shown how multiple sensors may imply that not the worst leaking sensors should be protected.}
Additionally, instead of injecting additional noise to improve privacy, a new approach where the defender makes use of the existing noise was presented. Using a convex program with a regularization term, the inputs at steady state could be made more private by increasing the regularization parameter. %{\color{blue}[informal... 'The regularizer', `'The regularization parameter', or `The penalty parameter']
The value of this parameter was shown to capture
%{\color{blue}shown to capture } 
% found to have an interpretation of 
the privacy-utility trade-off. Furthermore, it was shown that increasing the steady state cost in the more private direction is a sufficient but not necessary condition for decreasing the privacy leakage. Finally, these results were verified on numerical simulations.

{Connections between the Cram\'er-Rao lower bound and differential privacy have previously been discussed in \cite{farokhisandberg2019}. Since the Chapman-Robbins bound is a generalization of the Cram\'er-Rao bound, one would expect that similar connections exist between differential privacy and~\eqref{eq:bound}. In fact, differential privacy can be used to establish a looser lower bound, which will be explored further in future work.}

Process noise is another type of noise which affects the system, and thus is of interest for future work. The difficulty that arises in this setting is that the process noise affects multiple time steps, making the corresponding expression for~\eqref{eq:sum} much more complex. Another future research direction is to analyze the same situation under smoother changes of the input.
% , but instead considering a smoother profile. 
This alternative approach would provide a generalization of the main results, giving the defender another dimension of possible defense strategies, for example, by {choosing time-varying variance of the noise, $\Sigma_e$}. Finally, relaxing the assumption that the defender needs to know the noise model \textit{a priori} would enable these results to be more applicable and is thus of high importance for future work.

\vspace{-5pt}
\section*{Acknowledgments}
This research was funded by the Swedish Foundation for Strategic Research through the CLAS project (grant RIT17-0046).

\vspace{-10pt}
\bibliographystyle{unsrt}
\bibliography{bibliography}

%\newpage
%\section*{Appendix}
\newpage
\onecolumn
\appendix

\begin{proof}[Proof of Theorem~\ref{thm:main}]
The estimator, $\psi_u(Y,M)$, uses the measurements $Y$ and the model $M$ in order to estimate the change time $k^*$. Here, $k^*$ is a parameter which determines the probability distribution of $Y$. The minimum variance of the estimator is given by the Chapman-Robbins bound~\cite{chapmanrobbins1951},
\begin{equation*}
    \mathrm{Var}(\psi(Y,M) \big |k^*) \geq \sup \limits_{\tau \neq 0} \frac{\left( \mathbb{E}\left[ \psi(Y,M) \big| k^* + \tau \right] - \mathbb{E}\left[ \psi(Y,M) \big | k^* \right]  \right)^2}{\mathbb{E} \left[ \left( \frac{P\left(Y | k^*+\tau\right)}{P\left(Y | k^*\right)} -1 \right)^2 \bigg | k^*\right]},
\end{equation*}
where $\tau \in \mathbb{Z}$ and $P(Y|k^*)$ is the probability of obtaining measurements $Y$, conditioned on the change time is $k^*$. Evaluating the expectation in the denominator gives
\begin{equation} \label{eq:boundproof}
    \mathrm{Var}(\psi(Y,M) \big |k^*) \geq \sup \limits_{\tau \neq 0} \frac{\left( \mathbb{E}\left[ \psi(Y,M) \big| k^* + \tau \right] - \mathbb{E}\left[ \psi(Y,M) \big | k^* \right]  \right)^2}{\displaystyle\int \limits_{\mathbb{R}^N} \frac{P\left(Y | k^*+\tau\right)^2}{P\left(Y | k^*\right)} \mathrm{d}Y-1}.
\end{equation}

Since $\frac{P\left(Y | k^*+\tau\right)^2}{P\left(Y | k^*\right)} = \mathrm{e}^{2\log{P(Y_{}| k^*+\tau)}-\log{P(Y_{}| k^*)}}$, we write for $\tau \geq 1$,
\begin{align*}
    & \log{P(Y_{}| k^*+\tau)}-\log{P(Y_{}| k^*)} = \\ 
    & - \frac{1}{2}\sum \limits_ {k=1}^N {\left(y_k-CA^kx_0 - \sum \limits_{l = k^* + \tau}^{k-1}CA^{k-l-1}Bu \right)}^\top\Sigma_e^{-1}{\left(y_k-CA^kx_0 - \sum \limits_{l = k^* + \tau}^{k-1}CA^{k-l-1}Bu \right)} + \\
    & \frac{1}{2}\sum \limits_ {k=1}^N {\left(y_k-CA^kx_0 - \sum \limits_{l = k^*}^{k-1}CA^{k-l-1}Bu \right)^\top\Sigma_e^{-1}\left(y_k-CA^kx_0 - \sum \limits_{l = k^*}^{k-1}CA^{k-l-1}Bu \right)} =  \\
    & -\frac{1}{2} \sum \limits_{k=k^*+1}^N   (\sum \limits_{l = k^*}^{\mathrm{min}( k^* + \tau -1,k-1)} CA^{k-1-l}Bu)^\top\Sigma^{-1}_e(2y_k-(2CA^kx_0 + \sum \limits_{l = k^*}^{k-1}CA^{k-l-1}Bu + \sum \limits_{l = k^* + \tau}^{k-1}CA^{k-l-1}Bu)).
\end{align*}

Continuing, we see that,
\begin{align*} 
    & 2\log{P(Y_{}| k^*+1)}-\log{P(Y_{}| k^*)} = 2(\log{P(Y_{}| k^*+1)}-\log{P(Y_{}| k^*)})+\log{P(Y_{}| k^*)} = \\
    & -\frac{1}{2} \sum \limits_{k=k^*+1}^N   (\sum \limits_{l = k^*}^{\mathrm{min}( k^* + \tau -1,k-1)} CA^{k-1-l}Bu)^\top\Sigma^{-1}_e(2y_k-(2CA^kx_0 + \sum \limits_{l = k^*}^{k-1}CA^{k-l-1}Bu + \sum \limits_{l = k^* + \tau}^{k-1}CA^{k-l-1}Bu))- \\
    & - \frac{1}{2}\sum \limits_ {k=1}^N {\left(y_k-CA^kx_0 - \sum \limits_{l = k^*}^{k-1}CA^{k-l-1}Bu \right)^\top\Sigma_e^{-1}\left(y_k-CA^kx_0 - \sum \limits_{l = k^*}^{k-1}CA^{k-l-1}Bu \right)} =  \\
    & - \frac{1}{2}\sum \limits_{k=k^* +1}^N\left(y_k- G\right)^\top \Sigma^{-1}_e\left(y_k- G\right) + \sum \limits_{k=k^* +1}^N\left(\sum \limits_{l=k^*}^{\mathrm{min}( k^* + \tau -1,k-1)} CA^{k-1-l}Bu\right)^\top \Sigma_e^{-1}\left(\sum \limits_{l=k^*}^{\mathrm{min}( k^* + \tau -1,k-1)} CA^{k-1-l}Bu\right)  \\
    & - \frac{1}{2}\sum \limits_ {k=1}^{k^*} {\left(y_k-CA^kx_0 - \sum \limits_{l = k^*}^{k-1}CA^{k-l-1}Bu \right)^\top\Sigma_e^{-1}\left(y_k-CA^kx_0 - \sum \limits_{l = k^*}^{k-1}CA^{k-l-1}Bu \right)},
\end{align*}
where, 
\begin{equation*}
    G = CA^kx_0 + \sum \limits_{l=k^*}^{k-1}CA^{k-l-1}Bu- 2\left(\sum \limits_{l=k^*}^{\mathrm{min}( k^* + \tau -1,k-1)} CA^{k-1-l}Bu\right).
\end{equation*}
Inserting this expression into the bound in~\eqref{eq:boundproof}, evaluating the integral, and setting $\mathbb{E}\left[ \psi(Y,M) \big | k^* \right] = k^* + g(k^*)$, we obtain
\begin{equation}\label{eq:proofres}
    \mathrm{Var}(\psi(Y,M) \big |k^*) \geq \max \limits_{\tau \geq 1} \frac{ \left(\tau+g(k^*+\tau) - g(k^*)\right)^2    }{\mathrm{e}^{u^\top\mathcal{S}(\tau,M)u }-1},
\end{equation}
where,
\begin{equation*}
   \mathcal{S}(\tau,M) =  \sum \limits_{k=k^* +1}^N\left(\sum \limits_{l=k^*}^{\mathrm{min}( k^* + \tau -1,k-1)} CA^{k-1-l}Bu\right)^\top \Sigma_e^{-1}\left(\sum \limits_{l=k^*}^{\mathrm{min}( k^* + \tau -1,k-1)} CA^{k-1-l}Bu\right). 
\end{equation*}
For $\tau \leq -1$, an equivalent calculation can be made giving the same expression as~\eqref{eq:proofres}, but replacing $\mathcal{S}$ with
\begin{equation*}
    \mathcal{S}_-(\tau,M) = \sum \limits_{k=\hat k+\tau +1}^N\left(\sum \limits_{l=\hat k+\tau}^{\mathrm{min}(\hat k -1,k-1)} CA^{k-1-l}Bu \right)^\top \Sigma_e^{-1}\left(\sum \limits_{l=\hat k+\tau}^{\mathrm{min}(\hat k -1,k-1)} CA^{k-1-l}Bu \right).
\end{equation*}

However, since
\begin{equation*}
    \mathcal{S}(|\tau|,M) \leq \mathcal{S}_-(-|\tau|,M),
\end{equation*}
for each positive integer $|\tau|$, we can ignore the $\tau \leq -1$ cases.
\end{proof}

\end{document}